\documentclass[manuscript]{acmartinoc2026}
\def\BibTeX{{\rm B\kern-.05em{\sc i\kern-.025em b}\kern-.08em
    T\kern-.1667em\lower.7ex\hbox{E}\kern-.125emX}}

\usepackage{booktabs} 
\usepackage{adjustbox} 
\usepackage{subcaption}
\usepackage{siunitx}
\usepackage{tikz}
\usetikzlibrary{positioning}
\usetikzlibrary{calc}
\usetikzlibrary{fit}
\usetikzlibrary{arrows.meta}
\usepackage{algorithm}
\usepackage[noend]{algorithmic}

\begin{document}
\title{An improved Lower Bound for Local Failover in Directed Networks via Binary Covering Arrays}

\author{Erik van den Akker}
\orcid{0000-0003-4956-255X}
\affiliation{%
  \institution{TU Dortmund University}
  \country{Germany}
}
\email{erik.vandenakker@tu-dortmund.de}

\author{Klaus-Tycho Foerster}
\orcid{0000-0003-4635-4480}
\affiliation{%
  \institution{TU Dortmund University}
  \country{Germany}
}
\email{klaus-tycho.foerster@tu-dortmund.de}

\renewcommand{\shortauthors}{}
 
\setlength{\parindent}{0pt}

\begin{abstract}
Communication networks often rely on some form of local failover rules for fast forwarding decisions upon link failures. While on undirected networks, up to two failures can be tolerated, when just matching packet origin and destination, on directed networks tolerance to even a single failure cannot be guaranteed. Previous results have shown a lower bound of at least $\lceil\log(k+1)\rceil$ rewritable bits to tolerate $k$ failures. 

We improve on this lower bound for cases of $k\geq 2$, by constructing a network, in which successful routing is linked to the \textit{Covering Array Problem} on a binary alphabet, leading to a lower bound of $\Omega(k + \lceil\log\log(\lceil\frac{n}{4}\rceil-k)\rceil)$ for $k$ failures in an $n$ node network.

\end{abstract}



\maketitle
\section{Motivation and Background}

\paragraph{Motivation.} Fast failover routing enables communication networks to recover from link failures within microseconds~\cite{DBLP:journals/comsur/ChiesaKRRS21}, without waiting for relatively slow global control plane reconvergence~\cite{DBLP:conf/nsdi/LiuPSGSS13}. Such mechanisms are typically implemented locally at each node, using preinstalled forwarding rules that react only to the status of outgoing links and limited packet header information.

\paragraph{Related Work on Undirected Graphs.} 
Feigenbaum et al.~\cite{DBLP:conf/podc/FeigenbaumGPSSS12} have introduced the problem on undirected graphs and have shown that, given only the packet destination $t$, resilience against a single failure is always possible, while perfect resilience (against an arbitrary number of failures) is generally impossible. Chiesa et al.~\cite{DBLP:journals/ton/ChiesaNMGMSS17} improved this bound, showing that even resilience to two failures is not possible on all topologies.  Dai et al.~\cite{DBLP:conf/spaa/DaiF023} showed that, matching on the origin $s$ of the packet in the packet header as well, resilience against two failures is always possible, while resilience against three or more failures cannot be~guaranteed. Moreover, perfect resilience is even impossible on very restricted graph classes~\cite{DBLP:conf/dsn/FoersterHPST22}, though it can be decided in polynomial time for undirected graphs~\cite{bentert2026perfectnetworkresiliencepolynomial}. 

As it is still an open question if resilience against $k$ link failures can be achieved without matching on the packet source in $k$-edge-connected graphs, Chiesa et al.~\cite{DBLP:journals/ton/ChiesaNMGMSS17} proposed a method using $3$-bits in the packet headers, rewritable by routers, such that resilience against $k-1$ failures can be guaranteed in a $k$-edge-connected graph without matching on the source.

\paragraph{Related Work on Directed Graphs.} 
On directed networks, van den Akker and Foerster~\cite{DBLP:conf/sss/AkkerF24} have previously shown that even resilience against a single failure cannot be guaranteed on directed networks, giving a lower bound of $\lceil\log(k+1)\rceil$ needed rewritable header bits, where $k$ is the number of failed arcs in the network. While there is work on heuristic algorithms~\cite{DBLP:conf/rndm/GrobeAF24} without guarantees, we are not aware of any lower bounds that include the network size $n$.

\paragraph{Overview.} In \S\ref{sec:model} we describe the model used for directed failover. In \S\ref{sec:cover} we briefly describe the covering array problem. In \S\ref{sec:bound} we present our new contstruction for a lower bound linked to the covering array problem. Lastly, in \S\ref{sec:conc} we conclude.

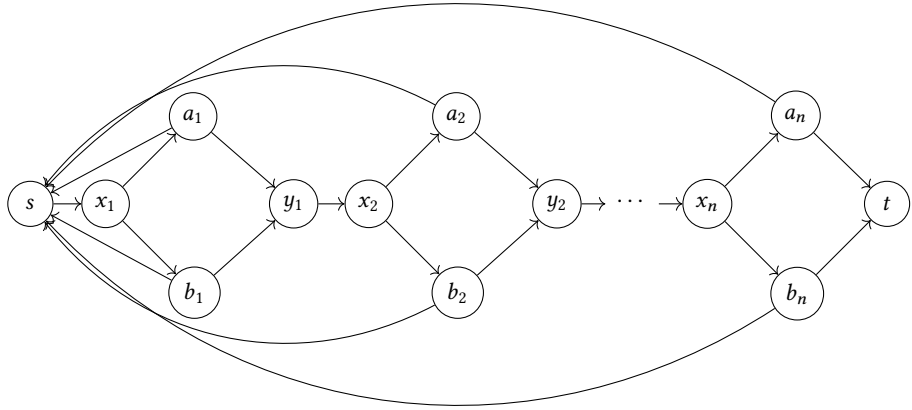
\begin{figure*}[t!]
    \centering
    \begin{tikzpicture}[->, node distance=1cm, 
        every node/.style={circle, draw, minimum size=6mm}]
        
        \node (S) {$s$};
        \node (X1) [right of=S] {$x_1$};
        \node (A1) [above right=of X1] {$a_1$};
        \node (B1) [below right=of X1] {$b_1$};
        \node (Y1) [right of=X1, xshift=1.5cm] {$y_1$};
        \node (X2) [right of=Y1] {$x_2$};
        \node (A2) [above right=of X2] {$a_2$};
        \node (B2) [below right=of X2] {$b_2$};
        
        \draw (S) -- (X1);
        
        \draw (X1) -- (A1);
        \draw (X1) -- (B1);
        \draw (A1) -- (Y1);
        \draw (B1) -- (Y1);
        \draw (Y1) -- (X2);
        
        \draw (X2) -- (A2);
        \draw (X2) -- (B2);
        \node (Y2) [right of=X2, xshift=1.5cm] {$y_2$};
        
        \node (dots) [right of=Y2, draw=none] {$\cdots$};
        \draw (A2) -- (Y2);
        \draw (B2) -- (Y2);
        \draw (Y2) -- (dots);
      
        \node (Xn) [right of=dots] {$x_n$};
        \node (An) [above right=of Xn] {$a_n$};
        \node (Bn) [below right=of Xn] {$b_n$};
        \node (Yn) [right=of Xn,  xshift=0.75cm] {$t$};
        
        \draw (dots) -- (Xn);
        \draw (Xn) -- (An);
        \draw (Xn) -- (Bn);
        \draw (An) -- (Yn);
        \draw (Bn) -- (Yn);

        \draw (A1) -- (S);
        \draw (B1) -- (S);

        \draw (A2) to[bend right=40] (S);
        \draw (B2) to[bend left=40] (S);

         \draw (An) to[bend right=40] (S);
        \draw (Bn) to[bend left=40] (S);
    \end{tikzpicture}
    \caption{Network for the lower bound construction. All $a_i$ and $b_i$ have an arc back to $s$. Any of the arcs $a_i \rightarrow x_{i+1}$ can fail (but not both at the same time), there needs to be at least one failover path that does not run into any failed link.}
    \label{fig:chain}
    \vspace{2mm}
\end{figure*}

\section{Model}\label{sec:model}
\paragraph{Network and Forwarding Model.} In the directed fast failover problem, a network is given as a directed graph $G = (V,E)$, where the nodes represent the routers and the arcs represent the links of the network. For each router, a set of forwarding rules $in \times f \times s \times t \times b \Rightarrow out, b'$ is installed, which determine the next hop and modification of a bitstring in packet headers, by looking at the incoming arc $in$ used by the packet, the source $s$ and target $t$ of the packet, as well as the current bitstring $b$ and the local set of failures of outgoing links $f$. The packet will then be forwarded via the outgoing arc $out$ and the bitstring in the packet header will be rewritten to be $b'$.

\paragraph{Deadend Avoidance.} Candidate pairs of source and target $s,t \in V$ are only considered, if the following condition holds: All nodes that are reachable from $s$ (including $s$) must have some directed path to $t$. A set of arcs $F \subset E$ is only allowed to fail, when in $G' = (V, E - F)$ this condition still holds for $s$ and $t$. This restriction prevents the creation of deadends, which would render the problem impossible to solve on general~topologies.

\paragraph{Goal.} The goal is to construct a set of forwarding rules for each node, so that for each valid pair of source $s$ and target $t$, packets are correctly forwarded in $G$ as well as each $G' = (V, E - F)$ for each arc set $F$ which can be removed without violating the reachability~condition.

\paragraph{Remark:} Note that upper bounds for this problem (with multiple arc failures) can be transferred to the undirected problem, by considering an undirected graph a directed graph with arcs in both directions, in which both arcs always have to fail simultaneously.

\section{The Covering Array Problem}\label{sec:cover}
\paragraph{Overview.} We next define the covering array problem, give an example and cite a lower bound we will later use for our problem.
We refer to Hartman and Raskin~\cite{DBLP:journals/dm/HartmanR04} for a more detailed introduction.
\paragraph{Definition.} A Covering Array $CA(N;t,l,v)$ is a $N (rows) \times l (columns)$ Array, where each possible $t$-tuple from an Alphabet of size $v$ has to appear in each possible $N\times t$ subarray. The goal is to discover the minimum $N$ with which a $N \times l$ Covering Array can be created. This minimum $N$ is called the Covering Array Number $CAN(t,l,v)$.

\paragraph{Example.} Consider the following array over the alphabet $\{0,1\}$:

\begin{verbatim}

                       0	0	0	0
                       0	1	1	1
                       1	0	1	1
                       1	1	0	1
                       1	1	1	0
                       
\end{verbatim}

Here, each pair of two columns picked, 
contains all possible tuples $(0,0),(0,1),(1,0),(1,1)$, i.e., it is a Covering Array~$CA(5;2,4,2)$.

It is known that $CAN(2,4,2) = 5$, so the previous example is an optimal covering array for this case\cite{CHOI20122958}.

\paragraph{Lower Bound.} 
For $CAN(t,l,v)$ a lower bound of $$v^{t-2}\frac{v}{2}\log{(l-t+1)}(1+o(1))$$ is known, see Sarkar et al.~\cite{DBLP:journals/mst/SarkarCBV18}, for which we can fix $v=2$ since we only cover binary arrays, resulting in a lower bound of $$2^{t-2}\log(l-t+1)(1+o(1))~.$$\\

\section{Lower bound construction}\label{sec:bound}

\paragraph{Overview.} We now present a new construction that establishes stronger lower bounds on the number of bits required in packet headers. In particular, this construction demonstrates that the previously assumed bound of $\lceil \log(k+1) \rceil$ for $k$ failures~\cite{DBLP:conf/sss/AkkerF24} can be improved, even for two failures. 

To this end, we first construct a family of lower bound chain networks and then show that we need to have a sizable number of bits present in the packet to navigate these networks toward the source under failures.

\paragraph{Lower Bound Chain Networks.} Consider the network depicted in Figure~\ref{fig:chain}, where packets needs to be routed from $s$ to $t$. At each node $x_i$, a forwarding decision must be made: the packet is either sent to $y_i$ via $a_i$ or via $b_i$. If a failure occurs along the chosen path, the packet is returned to the source $s$. Crucially, once the packet returns to $s$, the routers at any $x_i$ cannot retain any state information beyond the modified bitstring in the packet header. Since this construction can be extended to arbitrary length, we denote a graph with $l$ such decision points $x_i$ as a chain of length $l$. Note that for an $n$ node Network $l$ can be at most $\lceil\frac{n}{4}\rceil-1$.

\paragraph{Advice needed at the source.} When routing towards the target $t$, if a packet runs into a failed arc from some $a_i$ or $b_i$ to $y_i$ (or $t$), it is returned to $s$, using the only intact arc. Thus, to reach $t$, there must be some configuration of the bitstring while the packet is located at $s$, leading to a chain of routing decisions, that avoids all failures and routes the packet to $t$ without looping via $s$.

Given a bitstring of $h$ header bits, there are at most $2^h$ such configurations the packet can have at $s$. If an adversary looks at all $s-t$ paths resulting from these configurations as a chain of $a$ or $b$ decisions (e.g. $abaa$ for a chain with $l=4$) and finds $k$ columns where some $k$-tuple from $\{a,b\}$ is not covered, they can break the opposite of the associated arcs, since all possible paths resulting from all configurations use at least one of those to reach $t$.

This means, that to survive all possible failure sets, each possible $k$-tuple from $\{a,b\}$ has to appear in each set of $k$ columns, over the paths taken by each configuration from $s$. 

This construction maps exactly to the Covering Array Problem over an alphabet of size two, with $CAN(k,l,2)$ being the minimum amount of header configurations needed to survive all possible sets of failures, leading to a minimum required header size of $\lceil\log(CAN(k,l,2))\rceil$ bits. From this follows:

\begin{theorem}
For each number of failures $k$ there are networks of size at least $4k+5$, such that $\Omega(k + \lceil\log\log(\lceil\frac{n}{4}\rceil-k)\rceil)$ header bits are needed to forward packets from $s$ to $t$. 
\end{theorem}

\begin{proof}
We construct a network with $l > k$, by using the previously mentioned construction with at least $k+1$ decision points, requiring at least $4k+5$ nodes.

Since we know that this maps to a $CAN(k,l,2)$ and the lower bound for any $CAN(t,l,2)$ is $$2^{t-2}\log(l-t+1)(1+o(1))$$ and we have $t=k$ and $l=\lceil\frac{n}{4}\rceil-1$, we get a lower bound of $$\lceil\log(2^{k-2}\log(\lceil\frac{n}{4}\rceil-k)(1+o(1))\rceil \geq k + \lceil\log\log(\lceil\frac{n}{4}\rceil-k)\rceil~.$$ 
\end{proof}

\section{Conclusion}\label{sec:conc}
We have shown a new lower bound for the minimum amount of rewritable bits needed for local failover routing on directed graphs, by leveraging a connection between the construction and the covering array problem on an alphabet of size $2$, leading to a lower bound of $\Omega(k + \lceil\log\log(\lceil\frac{n}{4}\rceil-k)\rceil)$ for $k$ arc failures in an $n$ node network (with at least $4k+5$ nodes).

Future Work could study, how and if a construction could make use of larger alphabet sizes in context of covering arrays to further improve this lower bounds. Additionally the upper bound for more than one failure still is the trivial upper bound given by just memorizing all encountered failures in the packet headers and relatively large with $k * \log E$ bits, so it would be interesting to know if this upper bound can be improved.

\newpage
\bibliographystyle{ACM-Reference-Format}
\bibliography{literature}

%
% ...

\end{document}